\documentclass[conference]{IEEEtran}
\IEEEoverridecommandlockouts
\usepackage{flushend}

\usepackage{cite}
\ifCLASSINFOpdf
   \usepackage[pdftex]{graphicx}
\else
\fi
\usepackage{amsmath}
\usepackage{amssymb}
\usepackage{amsthm}
\usepackage{algorithmic}
\usepackage{algorithm}
\usepackage{url}
\newtheorem{theorem}{Theorem}[section]

\newtheorem{corollary}[theorem]{Corollary}
\newtheorem{proposition}[theorem]{Proposition}

\theoremstyle{definition}

\newtheorem{remark}{Remark}[section]
\newtheorem{definition}[theorem]{Definition}

\numberwithin{equation}{section}

\newcommand{\cT}{\mathcal{T}}
\newcommand{\cR}{\mathcal{R}}

\newcommand{\onev}{\mathbf{1}}

\newcommand{\R}{\mathbb{R}}

\newcommand{\E}{\mathbb{E}}
\newcommand{\bE}{\mathbb{E}}
\newcommand{\pP}{\mathbb{P}}
\newcommand{\bP}{\mathbb{P}}

\newcommand{\diag}{\operatorname{diag}}

\newcommand{\pred}{\operatorname{pred}}

\newcommand{\exponential}{\operatorname{Exponential}}

\DeclareMathOperator*{\argmin}{arg\,min}

\begin{document}
\title{Numerical Investigation of Metrics for Epidemic Processes on Graphs}

\author{\IEEEauthorblockN{Max Goering, Nathan Albin\\ and Pietro Poggi-Corradini}
\IEEEauthorblockA{Department of Mathematics\\
Kansas State University\\
Manhattan, KS 66506\\
Email: mlgoering@gmail.com, \\ \{albin,pietro\}@math.ksu.edu}
\and
\IEEEauthorblockN{Caterina Scoglio \\ and Faryad Darabi Sahneh}
\IEEEauthorblockA{Department of Electrical and\\Computer Engineering\\
Kansas State University\\
Manhattan, KS 66506\\
Email: \{caterina,faryad\}@ksu.edu}
\thanks{This material is based upon work supported by the National Science Foundation under Grant No. DMS-1201427 and No. DMS-1515810.}}


%Telephone: (800) 555--1212, Fax: (888) 555--1212}

\IEEEspecialpapernotice{(In Proceedings of 2015 Asilomar Conference on Signals, Systems, and Computers)}

\maketitle

\begin{abstract}
This study develops the epidemic hitting time (EHT) metric on graphs measuring the expected time an epidemic  starting at node $a$ in a fully susceptible network takes to propagate and reach node $b$. An associated EHT centrality measure is then compared to degree, betweenness, spectral, and effective resistance centrality measures through exhaustive numerical simulations on several real-world network data-sets. We find two surprising observations: first, EHT centrality is highly correlated with effective resistance centrality; second, the EHT centrality measure is much more delocalized compared to degree and spectral centrality, highlighting the role of peripheral nodes in epidemic spreading on graphs.
\end{abstract}


% \fi
\IEEEpeerreviewmaketitle

\section{Introduction}
Dynamics on graphs has long been a central research topic across many applied disciplines. Several graph related quantities have proven successful in studying different applications. In particular, the effective resistance metric appears to be an important tool for studying a variety of dynamics over graphs, including, but not limited to, random walks on graphs, electrical networks, Markov chains, and averaging networks \cite{Boyd2008Siam}. It comes as no surprise that effective resistance is important for all these dynamic processes because they are gradient driven processes. The effective resistance is closely related to the Laplacian matrix of the underlying graph. However, epidemic spreading dynamics is a {\it branching process} and behaves very differently from gradient driven dynamics. 

In this paper, we seek graph quantities that help describe epidemic dynamics. Centralities are frequently used to determine properties of the underlying topology of a network. In fact, comparing different centralities on the same network can be used to classify the network structure \cite{ronqui2014analyzing}. Herein, we compare common centralities as well as graph metrics to see how to best understand epidemic dynamics. We use many of the same real-world data sets as in \cite{ronqui2014analyzing} and conclude that surprisingly, regardless of the underlying network structure, numerics indicate that the effective resistance is the most relevant graph quantity to the epidemic spreading. A partial explanation is offered at the end of the article.

\section{Epidemic Hitting Time Metric for Graphs}

\subsection{The SI model}\label{ssec:si}
The SI Epidemic Model is a model where every interaction between an infected and susceptible node can lead the susceptible node to become infected at a rate $\beta$ called the \emph{infection rate}. This means, that if two people, say Alice and Bob interact, and at time $t$ Alice is infected while Bob is susceptible, then the probability that Alice infects Bob in the time interval $(t, t+h]$ equals $\beta h + o(h)$.
Further, in the SI epidemic model, we assume that infections occur independently and once someone becomes infected they remain infected forever. In particular, by independence, the probability that two separate infections occur during a time interval $(t,t+h]$ is $o(h)$.

In realistic models, these interactions are described by links of a contact network $G = (V,E)$.
In order to keep track of the infection, we introduce the state vector $\omega_t$, where $\omega_{t}(i)$ is the state at time $t$ of the $i$--th node in the network. If $\omega_{t}(i) = 0$ we say node $i$ is susceptible and if $\omega_{t}(i) = 1$ we say node $i$ is infected. In order to make the scenario with Alice and Bob rigorous, we let $N = |V|$ and $A = [A(i,j)]_{1 \le i,j \le N}$ be the adjacency matrix representing the network $G$. Then $A(i,j)= 1$ if node $i$ can be infected by node $j$ and zero otherwise.  
Finally, let $I_{t}=\{i: \omega_t(i)=1\}$ be the set of nodes that are infected by time $t$. Since at most one infection occurs during $(t,t+h]$, we can condition on the possible infection occurrences and obtain:

%\begin{align*} 
%\pP\left[ \ \omega_{t+h}(i) = 1 \mid \omega_{t}(i) = 0, \omega_{t}\  \right] 
%&= \sum_{\substack{ j \in I_{t} \\ j \sim i }}\left(\beta h + o(h) \right) \\
%& = \sum_{j \sim i} \omega_{t}(j)  \beta h + o(h)  \\
%& = \beta h  \sum_{j =1}^{N} A(i,j) \omega_{t}(j) + o(h) , \\
%\end{align*}

\begin{equation}\label{eq:markov}
\pP\left[ \ \omega_{t+h}(i) = 1 \mid \omega_{t}(i) = 0, I_{t}\  \right] 
 = \beta h  \sum_{j =1}^{N} A(i,j) \omega_{t}(j) + o(h) 
\end{equation}

Next we consider the process $|I_{t}|$ which counts the total number of infected nodes in the network at time $t$. 
We find the transition probabilities for $|I_{t}|$ by summing over all susceptible nodes $S_{t}=\{i:\omega_t(i)=0\}$.
In other words, writing $|I_{t}|=\omega_t\cdot \onev$ where $\onev=[1\cdots 1]^T$ is a vector of ones, we have that
for $h>0$ small, 
\begin{equation*}
\pP \left[\ |I_{t+h}| - |I_{t}| = 1 \mid I_{t}  \ \right] = \beta h \left( \omega_t^T A \onev -  \omega_t^T A \omega_t\right)+o(h)
\end{equation*}

Note that $A\onev=d$ where $d(j)$ is the degree of node $j$.
Write $D=\diag(d)$ for the diagonal matrix of the node degrees. Then,
\[
\omega_t^T A \onev = \omega_t^T d = \omega_t^T D \onev = \omega_t^T D \omega_t,
\]
where the last equality follows  since $\sum_{i,j}\omega_t(i)D(i,j)=\sum_{i}d(i)\omega_t(i)=\sum_{i}d(i)\omega_t(i)^2=\sum_{i,j}\omega_t(i)D(i,j)\omega_t(j)$.

So letting $L=D-A$ be the combinatorial Laplacian, we get that
\[
 \omega_t^T A \onev-\omega_t^T A \omega_t= \omega_t^T D \omega_t-\omega_t^T A \omega_t=\omega_t^T L \omega_t
\]
Then
\begin{equation}\label{eq:quadform}
\pP \left[\ |I_{t+h}| - |I_{t}| = 1 \mid I_{t}  \ \right] = \beta h \omega_t^T L \omega_t + o(h).
\end{equation}
We learned to use the Laplacian in this equation from \cite{vanmieghem}.

The  set $V$ splits into two subsets $S_t$ and $I_t$ and $\omega_t$ is the indicator function of the set $I_t$. This partition defines a subset of edges called the edge-boundary, $\partial I_t$, consisting of all the edges that connect a node in $I_t$ to a node in $S_t$. With this in mind,
the quadratic form 
\[
\omega_t^T L \omega_t = \sum_{e\in E} \left[(\nabla \omega_t)(e)\right]^2 =  \sum_{e\in \partial I_t} 1 = |\partial I_t|
\]
counts the number of edges in $\partial I_t$.

So equation (\ref{eq:quadform}) becomes
\begin{equation} \label{iaprob}
 \pP \left[\ |I_{t+h}| - |I_{t}| = 1 \mid I_{t}  \ \right] = \beta h|\partial I_t|+o(h).
\end{equation}

 Now we determine the probability that a susceptible node $i$ will be the next node infected after time $t$, given $\omega_{t}$ and given that an infection occurs. Using the definition of conditional probability,
\begin{align*}
&\pP \left[ \ \omega_{t+h}(i) - \omega_{t}(i) = 1 \mid I_{t},  |I_{t+h}| - |I_{t}| = 1 \ \right]  \\
&=\frac{\pP\left[\  \omega_{t+h}(i) - \omega_{t}(i) = 1,  |I_{t+h}| - |I_{t}| = 1 \mid I_{t} \  \right]}{\pP \left[\ |I_{t+h}| - |I_{t}| = 1 \mid I_{t}  \ \right]}  \\ 
&=\frac{ \beta h (A\omega_t)(i) + o(h) }{\beta h|\partial I_t|+o(h)} =\frac{ (A\omega_t)(i) + o(1) }{|\partial I_t|+o(1)} 
\end{align*}
where the last line follows from (\ref{eq:markov}) and (\ref{iaprob}).

Letting $h \downarrow 0$ yields, we see that
$I_t$ evolves by choosing an active edge in $\partial I_t$ uniformly at random, and then infecting the susceptible endpoint of the chosen edge.

 The {\it arrival times} $Y_{0}, Y_{1}, \dots $ of the SI epidemic $\omega_{t}$ are defined by $Y_{0} = 0$ and
\begin{equation*}
Y_{k} = \inf \{ t\geq 0 : |I_t| = k \} \quad \text{for } k = 1,2, \dots
\end{equation*}

The {\it interarrival times} $T_{1}, T_{2}, \dots$ are the times between successive arrivals,
\begin{equation*}
T_{k} = Y_{k} - Y_{k-1} \quad \text{for } k = 1,2, \dots
\end{equation*}
Therefore, given the set of infected nodes $I_{Y_k}$, the next  arrival time  $T_{k+1}$ satisfies
$$
T_{k+1} \sim \exponential( \beta | \partial I_{Y_{k}} |).
$$
More precisely,
\[
\pP\left[\ T_{k+1}\leq t \ \mid\  I_{Y_k}=I \ \right]=1-e^{-\beta|\partial I| t}.
\]
This gives rise to an event-based algorithm that picks an active edge uniformly and updates the clock exponentially based on the size of the active set.

\subsection{Epidemic Hitting Time and Variable-Lengths Models}

The SI model describes a basic epidemic process which can be thought of as a Markov chain on the set of subsets of $V$. It is natural to introduce the notion of epidemic hitting time. 
\begin{definition}
Given two nodes $a\neq b$, start an SI epidemic $\omega_t$ at $a$. Define the hitting time
$T_{a,b}=\inf\{t\geq 0: b\in I_t\}$.
Then, the {\it epidemic hitting time} from $a$ to $b$ is the expected hitting time
\[
\tau_{a,b}=\E(T_{a,b}).
\]
\end{definition}

The variable-lengths model  for the network $G = (V,E)$ consists of assigning i.i.d. lengths $X_{e} \sim \exponential( \beta)$ for every $e\in E$. We then let $d(a,b)$ be the shortest-path distance with respect to these random edge-lengths. 

\begin{theorem}\label{thm:variable-length}
For a contact network $G= (V,E)$ with variable i.i.d. edge-lengths $X_{e} \sim \exponential(\beta)$ for each $e \in E$, the continuous time process $Z_t= \{ v \in V : d(a,v) \le t \}$  is a Markov process that evolves like an SI epidemic.
\end{theorem}
\begin{remark}
Theorem \ref{thm:variable-length}  is folklore in some circle (We thank Brent Werness for pointing out section 1.2.1 of \cite{qle} where this remark is made). 
\end{remark}

\subsection{Properties of Epidemic Hitting Time}

\begin{theorem}\label{thm:ehtmetric}
The epidemic hitting time is a metric (it is non-degenerate, symmetric and  satisfies the triangle inequality).
\end{theorem}
\begin{proof}
Theorem \ref{thm:ehtmetric} follows directly from Theorem \ref{thm:variable-length}, since $\tau_{a,b}=\bE\left[d(a,b)\right]$ and the expected value of a metric is still a metric.
\end{proof}
We say that a graph $G_{1} = (V_{1}, E_{1})$ is a refinement of $G = (V,E)$ if $V \subset V_{1}$ and $E \subset E_{1}$. 
\begin{proposition} \label{refinement}
If $G_{1}$ is a refinement of $G$, then for all $a,b \in V$, the epidemic hitting time $\tau_{a,b}(G) \ge \tau_{a,b}(G_{1})$. 
\end{proposition}
\begin{proof}
By Theorem \ref{thm:variable-length} adding more edges can only shorten $d(a,b)$, and adding new vertices cannot lengthen the shortest walk between $a$ and $b$. So $Z_t(G)\subset Z_t(G_1)$.
\end{proof}

\begin{proposition} \label{prop:Tree} Assume $\beta=1$.
If $\cT$ is a tree, then the epidemic hitting time is equal to the graph metric $d_0$ (shortest-path metric using number of hops). 
Furthermore, for an arbitrary graph $G$, the epidemic hitting time is bounded above by $d_0$.
\end{proposition}

\begin{proof}
This also follows from Theorem \ref{thm:variable-length}. On a tree there is a unique simple path from $a$ to $b$. So the expectation of the distance $d(a,b)$ can be computed using linearity. That is, if $\gamma_{a,b}$ denotes the unique simple path starting at $a$ and ending at $b$, then
$$
\E(d(a,b)) =  \sum_{e \in \gamma_{a,b}} \E(X_e) =  |\gamma_{a,b}|= d_0(a,b).
$$ 
The `Furthermore' statement follows from Proposition \ref{refinement}.
\end{proof}

\begin{proposition} \label{prop:KN}
For the complete graph $K_{N}$ on $N$ nodes. Given $a\neq b$ and $\beta=1$,  the epidemic hitting time  equals 
\begin{equation}\label{EHT:KN}
H_N(a,b)= \frac{\log N}{N}+o(1).
\end{equation}
Therefore, for an arbitrary graph $G$, the epidemic hitting time is bounded below by the expression in \eqref{EHT:KN}. 
\end{proposition}

\begin{proof}
To compute $\tau_{a,b}$, we condition on the event $T_{a,b}=Y_k$, meaning that $b$ is the $k$th infected node.
\[
\tau_{a,b}=\sum_{k=1}^{N-1} \bE\left[ T_{a,b} \mid T_{a,b}=Y_k \right]\bP(T_{a,b}=Y_k)
\]
Note that by symmetry
\begin{equation*}
\bP(T_{a,b}=Y_k) =\frac{1}{N-1}.
\end{equation*}

Now, if $T_{a,b}=Y_k$ that means that the first interarrival time $T_1$ is distributed like the minimum of $N-1$ exponential variables, $T_2$ like the minimum over $2(N-2)$ variables, etc... and so $ \bE\left[ T_{a,b} \mid T_{a,b}=Y_k \right]$ should be equal to:
\[
\frac{1}{(N-1)}+\frac{1}{2(N-2)}+\cdots+\frac{1}{k(N-k)}
\]
With this we get
\begin{align*}
\tau_{a,b} &=\frac{1}{N-1}\sum_{k=1}^{N-1}\sum_{j=1}^k \frac{1}{j(N-j)}\\
& = \frac{1}{N-1}\sum_{j=1}^{N-1}\sum_{k=j}^{N-1} \frac{1}{j(N-j)}\\
& = \frac{1}{N-1}\sum_{j=1}^{N-1}\frac{1}{j}\asymp
\frac{\log N}{N}
\end{align*}
\end{proof}

Combining our refinement observation (Proposition \ref{refinement}) with Propositions \ref{prop:Tree} and \ref{prop:KN}, we have the following corollary.
\begin{corollary}
Given a contact graph $G = (V,E)$ and two nodes $a,b$ in $V$. Then 
$$
H_{N}(a,b)\le \tau_{a,b} \le d_0(a,b).
$$
\end{corollary}

A different lower bound was proved by Lyons, Pemantle and Peres in \cite{lyons1999resistance}.
\begin{theorem}[\cite{lyons1999resistance}]
Let $G=(V,E)$ be a contact network with variable i.i.d. edge-lengths $X_{e} \sim \exponential(1)$ for each $e \in E$. Let $a,b$ be two nodes in $V$. Then
\[
\tau_{a,b}=\E\left[d(a,b)\right]\geq \cR(a,b)
\]
where $\cR(a,b)$ is the effective resistance between $a$ and $b$ when  $G$ is an electrical network with edge-conductances equal to $1$ Ohm.
\end{theorem}
 Note that for a complete graph $\cR(a,b)$ is smaller than $H_N(a,b)$, while for a path graph it is larger. In general, the relation between these two bounds depends on the structure of the graph.

\subsection{The ``Dijkstra on the fly'' algorithm}
When computing the shortest-path distance $d(a,b)$, one may implement Dijkstra's algorithm. However, instead of determining all the edge-lengths $X_e$ at the beginning, one may compute the length of an edge as needed. So to obtain the
 ``Dijkstra on the fly'' algorithm from Dijsktra's algorithm it is enough to insert line 8. The advantage is that only the edges in the ball induced by $\{x\in V: d(x,a)\leq d(b,a)\}$ are needed.

\begin{algorithm}[h!]
	\caption{Dijkstra on the fly algorithm}
	\label{dijkstra-onfly}
	\begin{algorithmic}[1]
		\REQUIRE $(G, s)$
		\STATE $W \gets \emptyset$
		\STATE $\delta(s) \gets 0$
		\STATE $\delta(u) \gets \infty$ and $\pred(u)\gets \emptyset$ for all $u \in V\setminus\{s\}$.
		\WHILE{There exists $u \not\in W$ satisfying $\delta(u) < \infty$}
			\STATE  $u \gets \argmin \{ \delta(u) : u \not\in W\}$
			\STATE  $W \gets W \cup \{u \}$
                   \FOR{ Each $v \not\in W$ that neighbors $u$}
					\STATE $\ell(u,v)\gets\exponential(\beta)$
					\IF{ $\delta(u) + \ell(u,v) < \delta(v)$}
						\STATE $\delta(v) = \delta(u) + \ell(u,v)$
						\STATE $\pred(v) \gets u$.
					\ENDIF
			\ENDFOR
		\ENDWHILE
	\end{algorithmic}
\end{algorithm}

\section{Metrics and Centralities} \label{MTC} 

A {\it graph centrality} is a function $P : V \to \R$ that assigns a value to each node. This value represents how important a node is with respect to the graph topology. It is often convenient to normalize a centrality so that it can be thought of as a probability distribution.

A {\it graph metric} is a symmetric, non-degenerate function $d: V \times V \to [0,\infty]$ that satisfies the triangle inequality. Every graph metric induces a centrality denoted 
\begin{equation}\label{eq:pd}
P_{d}(v_i) := \frac{1}{N} \sum_{v_{j} \neq v_{i}} \frac{1}{d( v_{i}, v_{j})}.
\end{equation}

Heuristically, this induced centrality measures how close node $v_i$ is to the rest of the network. 

\subsection{Metrics}

\emph{Epidemic Hitting Time: } This is the expected time it takes for an infection starting at node $i$ to infect node $j$. Theorem \ref{thm:ehtmetric} establishes that the epidemic hitting time is actually a graph metric.
For numerical approximations, one can run an $SI$ simulation using the event-based algorithm described at the end of Section \ref{ssec:si}, or sample a variable-lengths instance and run Dijkstra. For single source/target computations the event-based approach is more economical based on the number of coin tosses. However, for all-pairs computations finding all Dijkstra shortest-paths for a given variable-length instance saves computational time.

\emph{Effective Resistance:}  Arises by considering a graph as an electrical circuit, where each edge has $1$ Ohm of resistance. The pairwise effective resistance, $\cR(i,j)$ turns out to be a graph metric, see \cite{KR93}, and can be computed using the pseudo-inverse of the combinatorial Laplacian:
\begin{equation}\label{eq:effres}  
\cR =  \diag(L^{\dag}) + \diag(L^{\dag})^{T} -  2 L^{\dag}
\end{equation}
Effective resistance has been shown to be closely related to many properties of random Markov processes \cite{AF02}, including escape probabilities and commuting times of random walks \cite{CRRST96}, as well as  recurrence/transience properties of random walks \cite{DS84}. 

\subsection{Centralities}
\emph{Spectral Centrality: } The spectral centrality ($S$) is based on the idea that the importance of a node depends on the importance of its neighbors.  Spectral centrality does well in characterizing simple dynamics like diffusion, and is the basis of the PageRank algorithm \cite{brin98}. It is defined as
$$
S(i) = v_i= \lambda^{-1} \sum_{j \sim i} v_{j},
$$
where $\lambda$ is the largest eigenvalue of the adjacency matrix and $v_j$ is the $j$th component of the corresponding eigenvector.

\emph{Degree Centrality: } The degree centrality ($D$) of a node is simply  the number of neighbors. Therefore it ignores the importance a node's neighbors and  will frequently rank the importance of a large number of nodes to be the same.

\emph{Betweenness Centrality: } The betweenness centrality ($B$) first introduced in \cite{bavelas48} and popularized by Freeman \cite{freeman77} is a method of measuring the importance of a node based on the fraction of shortest paths that it lies on. It measures the influence of a node over the spread of information through the network. It is defined as
$$
B(i) = N^{-1} \sum_{j \neq i \neq k} \frac{ \sigma_{jk}(i)}{\sigma_{jk}},
$$
where $\sigma_{jk}$ denotes the number of shortest paths between $j$ and $k$ and $\sigma_{jk}(i)$ is the number of shortest paths from $j$ to $k$ that visit node $i$. Betweenness centrality is good for determining bottlenecks.

\emph{Communicability Centrality: } The communicability centrality ($C$) is an adaptation of Freeman's betweenness centrality that takes into account all independent walks between two nodes, instead of just the geodesic paths. The pairwise communicability can be computed from the spectrum of the adjacency matrix \cite{estrada2008communicability} as 
$$
C(i) = \sum_{j=1}^{N} \left(v_{j,i}\right)^{2} e^{\lambda_{j}},
$$
where $\lambda_1 \ge \lambda_2 \ge ... \ge \lambda_N$ are the eigenvalues of the adjacency matrix, and $v_{k,\ell}$ is the $\ell$th element of the eigenvector corresponding to the $k$th eigenvector.

\section{Real-World Networks}
In our simulations, we have used several real-world networks, available for download at \url{http://ece.k-state.edu/netse/projects/sprojects/proj2-products.html}

\begin{figure}[h]
{\footnotesize
{\setlength{\tabcolsep}{5pt}
\centering
\begin{tabular}{| l | l | l | l | l| l | l | l | l| l| l| l|}   

\hline
{\footnotesize Network} & {\footnotesize Number}& {\footnotesize Number } & {\footnotesize Average  } & {\footnotesize Average} & {\footnotesize Diam-}\\ 

\ & \footnotesize of & \footnotesize of & \footnotesize Node & \footnotesize Clustering & eter \\
\ & \footnotesize Nodes &  \footnotesize Edges &\footnotesize  Degree & \footnotesize Coefficient & \\  \hline

Adjnoun \cite{newman2006finding} & 112 & 425 & 7.589 & 0.1569 & 5  \\ \hline
Dolphins \cite{lusseau2003bottlenose} &62&	159	&5.129	&	0.309&	8	 \\ \hline
Facebook  \cite{leskovec2012learning} &2106 &2915  &2.768&0.117&12 \\ \hline
Football \cite{girvan2002community} & 115&	615	&10.695	&0.407	&4	\\ \hline
GRQC  \cite{leskovec2007graph} &4158&13422&6.456&0.557&17 \\ \hline
Hep  \cite{newman2001structure} & 8361	&15751	&3.768		&0.329&	19	\\ \hline
Karate \cite{zachary1977information}& 34	&78&	4.588&	0.256	& 5\\ \hline
Lesmis \cite{knuth1993stanford} & 77& 	254	& 6.597&	0.499&	5 \\ \hline 
Netscience \cite{newman2006finding} & 1589&	2742	&3.451	&0.693	&17\\ \hline
Polblogs \cite{adamic2005political} & 1490& 	19090	&25.624&	0.226	&9 \\ \hline
Power \cite{watts1998collective} & 4941 &	6594	&2.669 &	0.103	&46 \\ \hline
\end{tabular}}}
\caption{Summary of Network Structure Properties}
\end{figure}

\section{Numerical Methods and Results}
For each network described above, the epidemic hitting time was approximated by averaging the results of 100 SI simulations. Then a corresponding centrality was created using (\ref{eq:pd}). These results were compared to the centrality obtained from  effective resistance  as well as the spectral, degree, betweenness, and communicability centralities. Comparisons were done in two ways.

\emph{Total variation distance of probability measures: } One comparison was to compute the total variation between two centralities after normalizing the centralities so that they can be interpreted as a probability measure on the nodes. Given two centralities $P$ and $Q$ this amounts to computing (see \cite{levin2009markov}),
$$
\frac{1}{2} \sum_{v \in V} \left| P(v) - Q(v) \right|.
$$

In particular, the total variation distance of probability measures is a value between zero and two, where zero means the distributions overlap completely, and two means they are mutually singular.

\emph{Spearman's Rank Correlation Coefficient: } The Spearman Rank coefficient takes values between negative one and one and should be interpreted like the Pearson correlation coefficient \cite{lehman2005jmp}.  
Below are the results of these comparisons.
\begin{figure}[h!] 
{\footnotesize
{\setlength{\tabcolsep}{9pt}
\centering
		\begin{tabular}{| l | l | l | l | l | l| l |}
			\hline
			{\bf Network} & {\bf ER }&  {\bf BC }& {\bf DC } & {\bf CC}  & {\bf SC }\\ \hline
			Adjoun & {\bf 0.9918} & 0.8514	 &0.9814 & 0.9732	& 0.9727 \\ \hline
			Dolphins & {\bf 0.9846} & 0.7483 & 0.8579 & 0.8533 & 0.8521 \\ \hline
			Facebook & {\bf 0.9603}	& 0.4960 & 0.6715 & 0.8559 & 0.8875 \\ \hline
			Football & {\bf 0.8725}	&0.6142	&0.6753&	0.3683	&0.6471 \\ \hline
			GRQC &  {\bf 0.9651} & 0.5181 & 0.7536 & 0.7969 & 0.7960 \\ \hline
			HEP & {\bf 0.9554}	& 0.5387& 	0.6956&	0.8834&	0.8976 \\ \hline
			Karate &  0.9125	& 0.7814 &	0.8540 &	0.9419 &	{\bf 0.9525} \\ \hline
			Lesmis & {\bf 0.9900} &	0.7076 &	0.9490 &	0.9643&	0.9288 \\ \hline
			Netscience & {\bf 0.9744} &	0.3800&	0.4543&	0.6389&	0.6583 \\ \hline
			Polblogs & {\bf 0.996}	&0.8824&	0.9933&	0.968&	0.9681 \\ \hline
			Power & {\bf 0.8277}	&0.3133	&0.3247	&0.4207	&0.5919 \\ \hline
		\end{tabular} }}
	\caption{A table of values of the Spearman Rank coefficient between each centrality and the epidemic hitting time centrality. We note that in every network except the Karate club
	network, the Effective Resistance outperforms every other centrality. We also remark that in \cite{ronqui2014analyzing}, there was a similar observation of the Karate club
	network not behaving like the other networks, likely due its being a small network with a few dominating nodes.}
\end{figure}

\begin{figure}[h!] 
{\footnotesize
{\setlength{\tabcolsep}{9pt}
\centering
	\begin{tabular}{| l | l | l | l | l | l| l |}
		\hline
		{\bf Network} & {\bf ER}&  {\bf BC }& {\bf DC } & {\bf CC}  & {\bf SC }\\ \hline
		Adjnoun & {\bf 0.0304} &	0.4834	& 0.2123 &	0.4377	& 0.2172 \\ \hline
		Dolphins &{\bf 0.0288}	& 0.4155 &	0.1631	& 0.3424 &	0.347 \\ \hline
		Facebook & {\bf 0.0308}	& 0.8277 &	0.4633 &	0.7698 &	0.5904 \\ \hline
		Football & {\bf 0.0080}	& 0.1701 &	0.0233&	0.1113 &	0.0681 \\ \hline
		GRQC &  {\bf 0.0485} & 0.6710 & 0.3498 & 0.9669 & 0.8935              \\ \hline
		Hep &  {\bf 0.0492} &	0.6296 &	0.2698 &	0.9633	 & 0.9635 \\ \hline
		Karate & {\bf 0.0259} &	0.5841 &	0.2181 &	0.3271 &	0.1617 \\ \hline
		Lesmis & {\bf 0.0296} &	0.6902 &	0.2363 &	0.5282 &	0.3490 \\ \hline
		Netscience & {\bf 0.0252}  &	0.7706 &	0.2452 &	0.6294	& 0.7331 \\ \hline
		Polblogs & {\bf 0.0214} &	0.6436 &	0.3734 &	0.6293 &	0.3875 \\ \hline
		Power &  {\bf 0.0433} &	0.7098 &	0.2220 & 	0.2972 &	0.9846  \\ \hline
	\end{tabular}}}
	\caption{A table of values of the total variation between each centrality and the epidemic hitting time centrality. We note that in 		every network, the Effective Resistance
	outperforms every other centrality. This is in large part due to the weight of the peripheral nodes. The effective resistance and 			epidemic hitting time assign a large enough
	weight to these peripheral nodes, while the other centralities underestimate their importance.} 
\end{figure}

\begin{figure*}
\hspace*{-.9in}
\includegraphics[trim={0 2.4in 0 0}, clip, width = 1.15\textwidth]{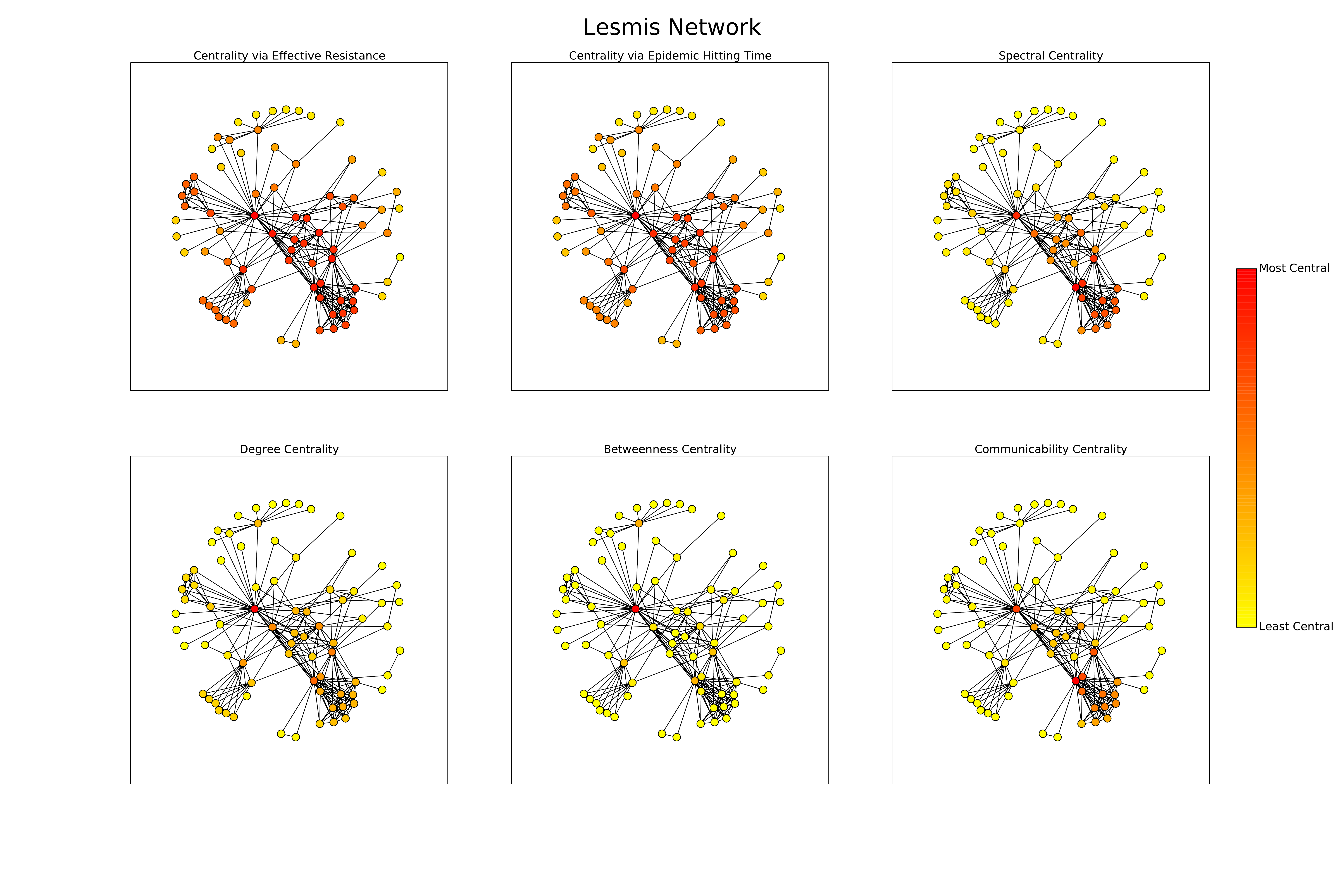}
\caption{Each figure shows the importance of the nodes in the Lesmis network, based off of a different centrality. We observe that not only do the effective resistance and epidemic hitting time centralities visually appear to match up quite well, they are also the only two centralities that assign relative importance to the peripheral nodes. See \cite{vsikic2013epidemic}.} 
\label{lesmis_heatmap}
\end{figure*}

This data indicates that effective resistance is a good measure of the influence that a node has over the spread of an epidemic throughout a contact network. However, the tables give little insight as to why effective resistance outperforms the other quantities. To gain intuition, we look at the heatmaps of the Lesmis network \cite{vsikic2013epidemic} in  Figure \ref{lesmis_heatmap} representing the importance of each node according to the respective graph quantity. We see that effective resistance and the epidemic hitting time are the only two graph quantities that assign relative importance to peripheral nodes of the network.

\section{Conclusion}
In this pape, we propose the epidemic hitting time (EHT) as a relevant metric on graphs to study epidemic processes. Epidemic hitting time between node $a$ and node $b$ measures the expected time it takes for an infection process starting at node $a$ in a fully susceptible network to propagate and reach node $b$. We develop the theory based on the susceptible-infected (SI) model, and show its equivalence with a variable link-length shortest-path model based on exponentially distributed random edge weights. As a result, we develop several efficient numerical methods for the EHT metric computation. Then, we conduct an exhaustive numerical experiment where we compute node centralities based on the EHT metric and compare the resulting ranking with the ones obtained using several other common node centralities including node degree, spectral, betweeness, communicability, and effective resistance centrality. We observe two surprising findings: first, EHT centrality is highly correlated with effective resistance centrality; second, EHT highlights the role of peripheral nodes in epidemic spreading unlike most other common centrality measures.

The variable-lengths model is useful to establish various properties of epidemic hitting time, but it also offers a partial explanation as to why our numerics show effective resistance as being closely correlated to the epidemic hitting time. In the work of Lyons, Pemantle, and Peres \cite{lyons1999resistance}, it was shown that the expected shortest path on a graph with i.i.d. exponentially distributed random lengths of mean one is bounded below by the effective resistance on the graph with unit resistance. In particular, this shows that the epidemic hitting time is bounded below by the effective resistance. This result confirms the recent finding of other authors (for example Sikic et al. \cite{vsikic2013epidemic}) that the impact of peripheral nodes is typically underestimated in epidemic models.
In summary, 
the computation of EHT-based node centrality provides novel 
information on the infection process on a graph, complementing the 
knowledge provided by common metrics.


%% use section* for acknowledgment
%\section*{Acknowledgment}



\bibliographystyle{IEEEtran}
%\bibliography{references.bib}

\end{document}